\RequirePackage{ifpdf}
\ifpdf

\else
\fail \fi

\RequirePackage[l2tabu, orthodox]{nag}
\RequirePackage{cmap}
\documentclass[runningheads,a4paper]{llncs}

\usepackage[english]{babel}
\usepackage{stackengine}
\usepackage[utf8]{inputenc}
\usepackage[T1]{fontenc}
\usepackage{lmodern}

\usepackage[protrusion=true,expansion=true]{microtype}
\usepackage{amsfonts,amssymb,xcolor}
\usepackage{stmaryrd}
\SetSymbolFont{stmry}{bold}{U}{stmry}{m}{n}
\usepackage[fixamsmath,disallowspaces]{mathtools}
\usepackage{fixmath}

\usepackage{fixltx2e}
\usepackage{scrhack}
\usepackage{relsize}
\usepackage{braket} 
\usepackage[   babel,              autostyle=true,strict
]{csquotes}
\MakeOuterQuote{"}

\usepackage[
  style=numeric,minalphanames=3, maxalphanames=4,         backend=bibtex,
  bibwarn=true,       texencoding=auto,   bibencoding=auto,   maxnames=5,
  maxbibnames=3,
  maxcitenames=5,
  url=false,
  isbn=false,
  doi=false,
  uniquelist
]{biblatex}

\usepackage[
bookmarks=false,
breaklinks=true,
pdfpagelayout=SinglePage
]{hyperref}

\usepackage[all]{hypcap}

\makeatletter
\g@addto@macro{\UrlBreaks}{\UrlOrds}
\makeatother

\usepackage[xspace]{ellipsis}

\usepackage[
  section    ]{placeins}

\usepackage{caption}[2011/08/06]

\hypersetup{
    colorlinks=true,
    linkcolor={red!50!black},
    citecolor={blue!50!black},
    urlcolor={blue!80!black}
}

\usepackage{pdfcolmk}

\makeatletter

\let\c@theorem\relax

\let\c@lemma\relax

\let\c@corollary\relax

\let\c@definition\relax

\let\c@example\relax

\makeatother

\makeatletter
\providecommand*{\cupdot}{  \mathbin{    \mathpalette\@cupdot{}  }}
\newcommand*{\@cupdot}[2]{  \ooalign{    $\m@th#1\cup$\cr
    \hidewidth$\m@th#1\cdot$\hidewidth
  }}
\makeatother

\usepackage{amsthm}
\usepackage{aliascnt}
\usepackage{cleveref} 

 \usepackage{booktabs}
\usepackage{colortbl}

\newcommand{\ie}{i.e.\@\xspace}
\newcommand{\eg}{e.g.\@\xspace}
\newcommand{\wrt}{w.\,r.\,t.\@\xspace}
\newcommand{\stt}{s.\,t.\@\xspace}
\newcommand{\suchthat}{s.\,t.\@\xspace}
\newcommand{\fa}{f.\,a.\@\xspace}

\newcommand{\Wloss}{W.l.o.g.\xspace}

      \newcommand{\complClFont}[1]{\mathbf{#1}}                                     \newcommand{\problemFont}[1]{\mathsf{#1}}         \newcommand{\mathCommandFont}[1]{\mathrm{#1}}     
\newcommand{\bigO}[1]{\protect\ensuremath{{\mathcal{O}(#1)}}}      
\newcommand{\size}[1]{{\protect\ensuremath{\vert\nobreak#1\nobreak\vert}}}

\newcommand{\leqlogm}{\protect\ensuremath{\leq^\mathCommandFont{log}_\mathCommandFont{m}}}

\newcommand{\ATIMES}[1]{\protect\ensuremath{\complClFont{ATIME}_\Sigma\left(#1\right)}\xspace}
\newcommand{\ATIMEP}[1]{\protect\ensuremath{\complClFont{ATIME}_\Pi\left(#1\right)}\xspace}
\newcommand{\ATIMEI}[2]{\protect\ensuremath{\complClFont{ATIME}_{#2}\left(#1\right)}\xspace}

\newcommand{\AP}{\protect\ensuremath{\complClFont{AP}}\xspace}
\renewcommand{\P}{\protect\ensuremath{\complClFont{P}}\xspace}
\newcommand{\NP}{\protect\ensuremath{\complClFont{NP}}\xspace}
\newcommand{\coNP}{\protect\ensuremath{\complClFont{coNP}}\xspace}

\newcommand{\PSPACE}{\protect\ensuremath{\complClFont{PSPACE}}\xspace}

\newcommand{\EXPSPACE}{\protect\ensuremath{\complClFont{EXPSPACE}}\xspace}

\newcommand{\EXP}{\protect\ensuremath{\complClFont{EXP}}\xspace}
\newcommand{\NEXP}{\protect\ensuremath{\complClFont{NEXP}}\xspace}
\newcommand{\AEXP}{\protect\ensuremath{\complClFont{AEXP}}\xspace}
\newcommand{\coNEXP}{\protect\ensuremath{\complClFont{coNEXP}}\xspace}
\newcommand{\AEXPPOLY}{\protect\ensuremath{\complClFont{AEXP}(\mathrm{poly})}\xspace}

\newcommand{\SigmaP}[1]{{\ensuremath\protect\Sigma^\mathCommandFont{P}_{#1}}}

\newcommand{\PiP}[1]{{\ensuremath\protect\Pi^\mathCommandFont{P}_{#1}}}
\newcommand{\DeltaP}[1]{{\ensuremath\protect\Delta^\mathCommandFont{P}_{#1}}}
\newcommand{\SigmaE}[1]{{\ensuremath\protect\Sigma^\mathCommandFont{E}_{#1}}}

\newcommand{\PiE}[1]{{\ensuremath\protect\Pi^\mathCommandFont{E}_{#1}}}
\newcommand{\DeltaE}[1]{{\ensuremath\protect\Delta^\mathCommandFont{E}_{#1}}}

\newcommand{\fr}{\ensuremath{\mathrm{Fr}}}

\newcommand{\QBF}{\protect\ensuremath\problemFont{QBF}}

\newcommand{\QBSF}{\protect\ensuremath\problemFont{QBSF}}
\newcommand{\SKOL}{\protect\ensuremath\problemFont{QBSF^{\mathrm{uniq}}}}

\newcommand{\calC}{\protect\ensuremath{\mathcal{C}}}
\newcommand{\calF}{\protect\ensuremath{\mathcal{F}}}

\newcommand{\calI}{\protect\ensuremath{\mathcal{I}}}

\newcommand{\calS}{\protect\ensuremath{\mathcal{S}}}
\newcommand{\calT}{\protect\ensuremath{\mathcal{T}}}

\providecommand{\dfn}{\mathrel{\mathop:}=}

\newtheoremstyle{theorem}
{\bigskipamount}{\medskipamount}{}{}{\bfseries}{.}{0.5em}{}
\newtheoremstyle{example}
{\bigskipamount}{\medskipamount}{}{}{\bfseries}{:}{\newline }{}
\newtheoremstyle{remark2}
{\bigskipamount}{\medskipamount}{}{}{\itshape}{:}{0.5em}{}
\newtheoremstyle{claimsty}
{\medskipamount}{\medskipamount}{}{}{\itshape}{.}{0.5em}{}

\theoremstyle{remark2}
\newtheorem*{remark}{Remark}

\theoremstyle{plain}

\newtheorem{theorem}{Theorem}
\newtheorem*{theorem*}{Theorem}

\newaliascnt{lemma}{theorem}
\newaliascnt{corollary}{theorem}
\newaliascnt{definition}{theorem}
\newaliascnt{example}{theorem}

\newtheorem{lemma}[lemma]{Lemma}
\newtheorem{corollary}[corollary]{Corollary}

\theoremstyle{definition}
\newtheorem{definition}[definition]{Definition}

\theoremstyle{theorem}
\newtheorem*{example}{Example}

\aliascntresetthe{lemma}
\aliascntresetthe{corollary}
\aliascntresetthe{definition}
\aliascntresetthe{example}

\crefname{theorem}{theorem}{theorems}
\crefname{lemma}{lemma}{lemmas}
\crefname{example}{example}{example}
\crefname{corollary}{corollary}{corollaries}
\crefname{definition}{definition}{definitions}
\crefname{example}{example}{examples}
 
\DeclareFieldFormat{labelnumberwidth}{\textcolor{darkgray}{\sffamily\bfseries#1}}

\DeclareFieldFormat[article]{title}{#1}

\DeclareFieldFormat[incollection]{title}{#1}

\DeclareFieldFormat[inproceedings]{title}{#1}

\DeclareFieldFormat[thesis]{title}{#1}

\renewbibmacro*{volume+number+eid}{  \printfield{volume}\space    \printtext[parens]{\usebibmacro{date}}
  \setunit{\addcomma\space}
  \printfield{number}  \setunit{\addcomma\space}  \printfield{eid}
  }

\renewbibmacro*{issue+date}{  \iffieldundef{issue}
      {}
      {\printfield{issue}}  \newunit}

\renewbibmacro{in:}{}
 \bibliography{exppl}

\begin{document}

\title{Complete Problems of Propositional Logic\\ for the Exponential Hierarchy}
\titlerunning{Complete Problems for the Exponential Hierarchy}

\author{Martin Lück}
\institute{Institut für Theoretische Informatik\\
Leibniz Universität Hannover, DE\\
\texttt{lueck@thi.uni-hannover.de}}
\authorrunning{M. Lück}

\maketitle

\begin{abstract}
Large complexity classes, like the exponential time hierarchy, received little attention in terms of finding complete problems.
In this work a generalization of propositional logic is investigated which fills this gap with the introduction of \emph{Boolean higher-order quantifiers} or equivalently \emph{Boolean Skolem functions}. This builds on the important results of Wrathall and Stockmeyer regarding complete problems, namely QBF and QBF$_k$, for the polynomial hierarchy. Furthermore it generalizes the \emph{Dependency QBF} problem introduced by Peterson, Reif and Azhar which is complete for $\NEXP$, the first level of the exponential hierarchy. Also it turns out that the hardness results do not collapse at the consideration of conjunctive and disjunctive normal forms, in contrast to plain QBF.
\end{abstract}

\section{Introduction}

The class of problems decidable in polynomial space, $\PSPACE$, can equivalently be defined as the class $\AP$, \ie, via alternating machines with polynomial runtime and no bound on the alternation number. The classes $\SigmaP{k}$ and $\PiP{k}$ of the polynomial hierarchy are then exactly the restrictions of $\AP$ to levels of bounded alternation \cite{alternation}. The problem of \emph{quantified Boolean formulas}, often called $\QBF$ resp. $\QBF_k$,  is complete for these classes. The subscript $k$ denotes the number of allowed quantifier alternations of a qbf in prenex normal form, whereas $\QBF$ imposes no bound on quantifier alternations. For this correspondence Stockmeyer called $\QBF$ the \emph{$\omega$-jump} of the bounded $\QBF_k$ variants, and similar the class $\PSPACE$ the $\omega$-jump of the polynomial hierarchy \cite{polyH}, in reference to the arithmetical hierarchy.

On the scale of exponential time the alternation approach leads to discrepancies regarding natural complete problems. Unbounded alternations in exponential time ($\AEXP$) leads to the same class as exponential space, in symbols $\AEXP = \EXPSPACE$, and therefore $\EXPSPACE$ is analogously the $\omega$-jump of exponential time classes with bounded alternations \cite{alternation}. Complete problems for $\EXPSPACE$ are rare, and often artificially constructed, frequently just succinctly encoded variants of $\PSPACE$-complete problems
\cite{allender_eric_minimum_2015,goos_second_1995, MeyerS72}.
If the number of machine alternations is bounded by a polynomial then this leads to the class $\AEXPPOLY$ which in fact lies between the exponential time hierarchy and its $\omega$-jump.

In this paper a natural complete problem is presented, similar to $\QBF_k$, which allows quantification over Boolean functions and is complete for the levels of the exponential time hierarchy.

The first appearance of such Boolean formulas with quantified Boolean functions was in the work of Peterson, Reif and Azhar who modeled games of imperfect information as a problem they called \emph{DQBF} or \emph{Dependency QBF} \cite{dqbf}. The basic idea is that in a game the player $\exists$ may or may not see the whole state that is visible to player $\forall$, and hence her next move must depend only on the disclosed information.
The existence of a winning strategy in a game can often be modeled as a formula with a prefix of alternating quantifiers corresponding to the moves. This naturally fits games where all information about the state of the game is visible to both players, as quantified variables always may depend on each previously quantified value. Any existentially quantified proposition $x$ can equivalently be replaced by its \emph{Skolem function} which is a function depending on the $\forall$-quantified propositions to the left of $x$.

To model imperfect information in the game, all one has to do is now to restrict the arguments of the Skolem function. For first-order predicate logic several formal notions have been introduced to accommodate this semantics, \eg, Henkin's branching quantifiers (see \cite{blass_henkin_1986}) or Hintikka's and Sandu's \emph{Independence Friendly Logic} \cite{hintikka_informational_1989}.

\subsubsection*{Contribution.}

The presented problem is a generalization of the QBF problem where Skolem functions of variables are explicit syntactical objects. This logic will be called QBSF as in \emph{Quantified Boolean Second-order Formulas}. It is shown that this introduction of function quantifiers to QBF (reminding of the step from first-order predicate logic to second-order predicate logic) yields enough expressive power the encode alternating quantification of exponentially large words. The problem of deciding the truth of a given QBF with higher-order quantifiers is complete for the class $\AEXPPOLY$, but has natural complete fragments for every level of the exponential hierarchy. 

The complexity of the problem is classified for several fragments, namely bounded numbers of function quantifiers and proposition quantifiers, as well as the restriction to formulas where function variables only occur as the Skolem functions of quantified propositions, \ie, always with the same arguments. The latter fragment is used as an alternative hardness proof for the original DQBF problem by Peterson, Reif and Azhar \cite{dqbf}.
 
\allowdisplaybreaks

\section{Preliminaries}

The reader is assumed to be familiar with usual notions of Turing machines (TMs) and complexity classes, especially in the setting of alternation introduced by Chandra, Kozen and Stockmeyer \cite{alternation}. In accordance to the original definition of alternating machines (ATMs) we distinguish them by the type of their initial state.
We abbreviate alternating Turing machines that start in an existential state as \emph{$\Sigma$ type machines} ($\Sigma$-ATMs), and those which start in a universal state as \emph{$\Pi$ type machines} ($\Pi$-ATMs).

We define $\EXP$ and $\NEXP$ as the classes of problems which are decidable by a \mbox{(non-)}deterministic machine in time $2^{p(n)}$ for a polynomial $p$.

\begin{definition}
For $Q \in \{ \Sigma, \Pi \}$ $g(n) \geq 1$, define $\ATIMEI{t(n), g(n)}{Q}$ as the class of all problems $A$ for which there is a $Q$-ATM deciding $A$ in time $\bigO{t(n)}$ with at most $g(n)$ alternations.
\end{definition}

The number of alternations is the maximal number of \emph{transitions} between universal and existential states or vice versa that $M$ does on inputs of length $n$, counting the initial configuration as first alternation. A polynomial time $\Sigma$-ATM ($\Pi$-ATM) with $g(n)$ alternations is also called $\SigmaP{g}$-machine ($\PiP{g}$-machine). For exponential time, \ie, $2^{p(n)}$ for al polynomial $p$, we analogously write $\SigmaE{g}$ resp. $\PiE{g}$.

\begin{definition}[\cite{alternation}]
\begin{alignat*}{2}
&\AEXP &&\dfn \bigcup_{t \in 2^{n^{\bigO{1}}}} \ATIMES{t,t}\text{,}\\
&\AEXPPOLY \;&&\dfn \bigcup_{\substack{t \in 2^{n^{\bigO{1}}}\\ p \in n^{\bigO{1}}}} \ATIMES{t,p}\text{.}
\end{alignat*}
\end{definition}

In this work we further require the notion of \emph{oracle Turing machines}. An oracle Turing machine is an ordinary Turing machine which additionally has access to an \emph{oracle language} $B$. The machine queries $B$ by writing an instance $x$ on a special \emph{oracle tape} and moving to a \emph{query state} $q_?$. But then instead of $q_?$ itself, one of two states, say, $q_+$ and $q_-$, is assumed instantaneously to identify the answer if $x \in B$ or not. There is no bound on the number of queries during a computation of an oracle machine, \ie, the machine can erase the oracle tape and ask more questions.

If $B$ is a language, then the usual complexity classes $\P, \NP, \NEXP$ etc. are generalized to $\P^B, \NP^B, \NEXP^B$ etc.\ where the definition is just changed from ordinary Turing machines to corresponding oracle machines with oracle $B$.
If $\mathcal{C}$ is a class of languages, then $\P^{\mathcal{C}} \dfn \bigcup_{B \in \mathcal{C}} \P^B$ and so on.

\medskip

To classify the complexity of the presented decision problems we require some standard definitions.

\begin{definition}
A \emph{logspace-reduction} from a language $A$ to a language $B$ is a function $f$ that is computable in logarithmic space such that $x \in A \Leftrightarrow f(x) \in B$.
If such $f$ exists then write $A \leqlogm B$. Say that $B$ is \emph{$\leqlogm$-hard} for a complexity class $\calC$ if $A \in \calC$ implies $A \leqlogm B$, and $B$ is \emph{$\leqlogm$-complete} for $\calC$ if it is $\leqlogm$-hard for $\calC$ and $B \in \calC$.
\end{definition}

\begin{definition}[The Polynomial Hierarchy \cite{polyH}]
The levels of the polynomial hierarchy are defined inductively as follows, where $k \geq 1$:
\begin{itemize}
\item $\SigmaP{0} = \PiP{0} = \DeltaP{0} \dfn \P$.
\item $\SigmaP{k} \dfn \NP^{\SigmaP{k-1}}$, $\PiP{k} \dfn \coNP^{\SigmaP{k-1}}$, $\DeltaP{k} \dfn \P^{\SigmaP{k-1}}$.
\end{itemize}
\end{definition}

\begin{definition}[The Exponential Hierarchy \cite{Mocas, Simon}]
The levels of the exponential hierarchy are defined inductively as follows, where $k \geq 1$:
\begin{itemize}
\item $\SigmaE{0} = \PiE{0} = \DeltaE{0} = \EXP$.
\item $\SigmaE{k} \dfn \NEXP^{\SigmaP{k-1}}$, $\PiE{k} \dfn \coNEXP^{\SigmaP{k-1}}$, $\DeltaE{k} \dfn \EXP^{\SigmaP{k-1}}$.
\end{itemize}
\end{definition}

\begin{theorem}[\cite{alternation}]\label{thm:ph_by_alternations}
For all $k \geq 1$:
\begin{align*}
\SigmaP{k} &= \bigcup_{p \in n^{\bigO{1}}} \ATIMES{p, k}\text{,}\\
\PiP{k} &= \bigcup_{p \in n^{\bigO{1}}} \ATIMEP{p, k}\text{.}
\end{align*}
\end{theorem}

The next two lemmas characterize the classes of the exponential hierarchy similar to the characterization of the polynomial hierarchy in \cite{polyH,Wrathall}. The proofs are rather straightforward adaptions of the characterization of the polynomial hierarchy.

First, it is possible to reduce a language recognized in alternating exponential time down to a language with deterministic polynomial time complexity by introducing additional \emph{word quantifiers}. These words roughly correspond to the "choices" of an encoded alternating machine, hence for the polynomial hierarchy words of polynomial length are quantified.
To encode machines deciding problems in $\SigmaE{k}$ or $\PiE{k}$ we require, informally spoken, \emph{large word quantifiers}, \ie, we quantify words of exponential length \wrt to the input.

\begin{lemma}\label{thm:eh_by_quantifiers}
For $k\geq 1$, $A \in \SigmaE{k}$ if and only if there is $t \in 2^{n^{\bigO{1}}}$\!\! and $B \in \P$ \stt{}
\[
x \in A \Leftrightarrow \exists y_1 \forall y_2 \ldots \Game_k y_k \; : \; \langle x,y_1,\ldots,y_n\rangle \in B\text{,}
\]
where $\Game_k = \forall$ for even $k$ and $\Game_k = \exists$ for odd $k$, and all $y_i$ have length bounded in $t(\size{x})$.
\end{lemma}
\begin{proof}
We have to show that for all $k\geq 1$ it is $A \in \SigmaE{k}$ if and only if there is $t \in 2^{n^{\bigO{1}}}$, $B \in \P$ \stt{}
\[
x \in A \Leftrightarrow \exists y_1 \forall y_2 \ldots \Game_k y_k \; : \; \langle x,y_1,\ldots,y_n\rangle \in B\text{,}
\]
where $\Game_k = \forall$ for even $k$ and $\Game_k = \exists$ for odd $k$, and all $y_i$ have length bounded in $t(\size{x})$.

\medskip

"$\Leftarrow$": Define
\[
D \dfn \Set{ \langle 0^{t(\size{x})},x,y_{1}\rangle |  \forall y_2 \ldots \Game_k y_k \; : \; \langle x,y_1,\ldots,y_n\rangle \in B }\text{,}
\]
where $0^{t(\size{x})}$ is the string consisting of $t(\size{x})$ zeros and the quantified $y_i$ are length-bounded by $t(\size{x})$. Then $D \in \PiP{k-1}$, and the algorithm that guesses a $y_1$ of length $\leq t(\size{x})$ and queries $D$ as oracle witnesses that $A \in \NEXP^{\PiP{k-1}} = \SigmaE{k}$.

"$\Rightarrow$":
Let $A$ be decided by some non-deterministic Turing machine $M$ with oracle $C \in \SigmaP{k-1}$. Assume that $M$ has runtime $t(n) = 2^{p(n)}$ for some polynomial $p$.
Consider now words of the form $z = \langle d,q,a\rangle$ of length $\bigO{t^2(\size{x})}$ where $d$ encodes $t(n)$ non-deterministic choices in a computation of $M$, $q$ encodes the oracle questions asked, and $a$ encodes the answers used by $M$. Then $x \in A$ if and only if there is such a word $z = \langle d,q,a\rangle$ \stt{} $M$ accepts on the computation encoded by the choices $d$, and $a$ are actually the correct answers of the oracle $C$ to the queries in $q$.

With given $\langle x,z\rangle = \langle x,d,q,a\rangle$ the encoded computation of $M$ on the path $d$ can be simulated deterministically in time polynomial in $\size{z}$. With given $\langle x,z\rangle$, also the problem of determining whether the answers $a$ for the queries $q$ are correct for the oracle $C$ is in $\P^C \subseteq \P^{\SigmaP{k-1}} \subseteq \SigmaP{k}$. Therefore the set of all tuples $\langle x,z\rangle$ which fulfill both properties, call it $C'$, is in $\SigmaP{k}$.
By the quantifier characterization of $\SigmaP{k}$ it holds that $\langle x,z\rangle \in C'$ if and only if $\exists y_1 \ldots \Game_k y_k \; : \; \langle x,z,y_1,\ldots,y_k\rangle \in B$ for some set $B \in \P$ and polynomially bounded, alternating quantifiers \cite{polyH,Wrathall}.
But then $x \in A \Leftrightarrow \exists \langle z,y_1\rangle \forall y_2 \ldots \Game_k y_k :  \langle x,z,y_1,y_2,\ldots,y_k\rangle \in B$ quantifiers with length bounded exponentially in $\size{x}$.
\end{proof}

We next state the known correspondence between the classes of the exponential hierarchy (which are defined via oracle machines) and the alternating time classes by the following lemma. It can be seen as the exponential equivalent of \Cref{thm:ph_by_alternations}.

\begin{lemma}\label{thm:alternating_exp_classes}
For all $k \geq 1$:
\begin{align*}
\SigmaE{k} &= \bigcup_{t \in 2^{n^{\bigO{1}}}} \ATIMES{t, k}\text{,}\\
\PiE{k} &= \bigcup_{t \in 2^{n^{\bigO{1}}}} \ATIMEP{t, k}\text{.}
\end{align*}
\end{lemma}
\begin{proof}
We show only the $\SigmaE{k}$ case as it can easily be adapted to the $\PiE{k}$ case.
For \enquote{$\subseteq$}, apply the foregoing \Cref{thm:eh_by_quantifiers}. Use an alternating machine to guess the exponentially long quantified words and check in deterministic exponential time if the resulting word is in $B$. Now to \enquote{$\supseteq$}. Let $t \in 2^{n^{\bigO{1}}}$ \stt{} $A \in \ATIMES{t,k}$. Then
$B \dfn \Set{\langle x,0^{t(\size{x})}\rangle | x \in A }$ is in $\SigmaP{k}$, therefore
\begin{align*}
x \in A &\Leftrightarrow \exists y_1 \ldots \Game_k y_k \; : \; \langle x,0^{t(\size{x})},y_1,\ldots,y_k\rangle \in C\\
&\Leftrightarrow \exists \langle y_0,y_1\rangle \forall y_2 \ldots \Game_k y_k \; : \; (y_0 = 0^{t(\size{x})}) \text{ and } \langle x,y_0,y_1,\ldots,y_k\rangle \in C\\
&\Leftrightarrow \exists \langle y_0,y_1\rangle \forall y_2 \ldots \Game_k y_k \; : \; \langle x,y_0,y_1,\ldots,y_k\rangle \in C'\\
\end{align*}
 for some $C,C' \in \P$ and alternating quantifiers which are exponentially bounded in $\size{x}$. By the previous lemma then it holds $A \in \SigmaE{k}$.
\end{proof}

Orponen gave a characterization of the exponential hierarchy via an \emph{indirect simulation} technique \cite{orponen}. He introduced it primarily due to its non-relativizing nature (while \emph{direct simulation} relativizes), however it also allows to use polynomial time machines to characterize languages with much higher complexity. Informally spoken, the whole computation of an exponential time machine is encoded into quantified oracles (instead of exponentially long words), which are then verified bit for bit, but in parallel, by an alternating oracle machine with only polynomial runtime. Baier and Wagner \cite{baier_analytic_1998} investigated more generally so-called \emph{type 0, type 1} and \emph{type 2} quantifiers, improving Orponen's result. These oracle characterizations play a major role for classifying the complexity of QBSF, the logic introduced in this paper, as they translate to the quantification of Boolean functions (which are per se exponentially large objects).

\begin{theorem}[\cite{baier_analytic_1998}]\label{thm:simulation}
Let $Q \in \{\Sigma, \Pi \}$, $k \in \mathbb{N}$. For every $L \in Q^E_{k}$ there is a polynomial $p$ and a deterministic polynomial time oracle machine $M$ \stt{}
\begin{align*}
x \in L \Leftrightarrow \; &\Game_1 A_1 \subseteq \{0,1\}^{p(\size{x})} \; \ldots\; \Game_{k} A_{k} \subseteq \{0,1\}^{p(\size{x})} \; \Game_{k+1} y \in \{0,1\}^{p(\size{x})}\\
&\text{\stt{}} \; M \text{ accepts }\langle x,y\rangle\text{ with oracle }\langle A_1, \ldots, A_{k}\rangle\text{,}
\end{align*}
where $\Game_1 = \exists$ if $Q = \Sigma$ and $\Game_1 = \forall$ if $Q = \Pi$, and $\Game_i \in \{ \exists, \forall \} \setminus \{ \Game_{i-1} \}$ for $1 < i \leq k+1$.
\end{theorem}

For sets $A_1, \ldots, A_k$ the term $\langle A_1, \ldots, A_k\rangle \dfn \Set{ \langle i, x\rangle | x \in A_i, 1 \leq i \leq k }$ is called \emph{efficient disjoint union} in this context. It allows the machine to access an arbitrary number of oracles in its computations by writing down the corresponding oracle index together with the query.

The following is a variant of the above theorem where the number $k$ of alternations is not fixed but polynomial in the input size:

\begin{theorem}[\cite{hannula_complexity_2015}]\label{thm:simulationpoly}
For every set $L \in \AEXPPOLY$ there is a polynomial $p$ and a deterministic polynomial time oracle machine $M$ \stt{}
\begin{align*}
x \in L \Leftrightarrow \;&\Game_1 A_1 \subseteq \{0,1\}^{p(\size{x})} \; \ldots\; \Game_{p(\size{x})} A_{p(\size{x})} \subseteq \{0,1\}^{p(\size{x})}\\
&\Game_1 y_1 \,\,\in \{0,1\}^{p(\size{x})} \; \ldots\; \Game_{p(\size{x})} y_{p(\size{x})} 	\,\,\in \{0,1\}^{p(\size{x})}\\
&\text{\stt{}} \; M \text{ accepts }	\langle x, y_1, \ldots, y_{p(\size{x})}\rangle \text{ with oracle }\langle A_1, \ldots, A_{p(\size{x})}\rangle\text{,}
\end{align*}
where $\Game_1 \ldots \Game_{p(\size{x})}$ is an alternating quantifier sequence.
\end{theorem}

Obviously each quantified word can be efficiently encoded in its own additional oracle. There are only polynomially many quantified words, so in the unbounded case we can drop the word quantifiers completely.

These characterizations all have tight upper bounds. Suppose that a language is characterized by such a sequence of quantified oracles. Then conversely an alternating machine can non-deterministically guess the oracle sets with runtime exponential in $p$ and then simulate $M$ including the word quantifiers in deterministic exponential time.
Together with \Cref{thm:alternating_exp_classes} we obtain:

\begin{corollary}\label{thm:alt_seq_poly}
$L \in \AEXPPOLY$ if and only if there is a polynomial $p$ and an deterministic polynomial time oracle machine $M$ \stt{} $x \in L$ iff
\begin{align*}
\Game_1 A_1 \subseteq \{0,1\}^{p(\size{x})} \;\ldots\; &\Game_{p(\size{x})} A_{p(\size{x})} \subseteq \{0,1\}^{p(\size{x})}\; :\;\\
&M \text{ accepts } x \text{ with oracle } \langle A_1, \ldots, A_{p(\size{x})}\rangle\text{,}
\end{align*}
where $\Game_1, \ldots, \Game_{p(\size{x})}$ is an alternating sequence of quantifiers.
\end{corollary}

\begin{corollary}
For all $k \geq 1$, $L \in \SigmaE{k}$ if and only if there is a polynomial $p$ and a deterministic polynomial time oracle machine $M$ \stt{} $x \in L$ iff
\begin{align*}
\exists A_1 \subseteq \{0,1\}^{p(\size{x})} \;\ldots\; &\Game_{k} A_{k} \subseteq \{0,1\}^{p(\size{x})}\; \Game_{k+1} y \in \{0,1\}^{p(\size{x})}\;:\;\\
&M \text{ accepts } \langle x,y\rangle \text{ with oracle } \langle  A_1, \ldots, A_{k}\rangle\text{,}
\end{align*}
where $\Game_1 = \exists, \ldots, \Game_{k+1}$ is an alternating sequence of quantifiers.
\end{corollary}

\begin{corollary}\label{thm:alt_seq_pik}
For all $k \geq 1$, $L \in \PiE{k}$ if and only if there is a polynomial $p$ and a deterministic polynomial time oracle machine $M$ \stt{} $x \in L$ iff
\begin{align*}
\forall A_1 \subseteq \{0,1\}^{p(\size{x})} \;\ldots\; &\Game_{k} A_{k} \subseteq \{0,1\}^{p(\size{x})}\; \Game_{k+1} y \in \{0,1\}^{p(\size{x})}\;:\;\\
&M \text{ accepts } \langle x,y\rangle \text{ with oracle } \langle  A_1, \ldots, A_{k}\rangle\text{,}
\end{align*}
where $\Game_1 = \forall, \ldots, \Game_{k+1}$ is an alternating sequence of quantifiers.
\end{corollary}
 
\section{Second-order QBF}

In this section the logic QBSF is introduced formally. It is a straightforward generalization of usual QBF to include function variables; it could be interpreted as a "second-order" extension: It behaves similarly to plain QBF as second-order logic behaves to first-order logic.

\begin{definition}[Syntax of QBSF]
The constants $1$ and $0$ are \emph{quantified Boolean second-order formulas (qbsfs)} . If $f^n$ is a function symbol of arity $n \geq 0$ and $\varphi_1, \ldots, \varphi_n$ are qbsfs, then $f^n(\varphi_1, \ldots, \varphi_n)$, $\varphi_1 \land \varphi_2$, $\neg \varphi_1$ and $\exists f^n \varphi_1$ are all qbsfs.
\end{definition}

Abbreviations like $\varphi \lor \psi$, $\varphi \rightarrow \psi$, $\varphi \leftrightarrow \psi$ and $\forall f^{n} \,\psi$ can be defined from this as usual. In this setting, propositions can be understood as functions of arity zero. If the arity of a symbol is clear or does not matter we drop the indicator from now on. Furthermore a sequence $x_1, \ldots, x_s$ of variables can be abbreviated as $\vec{x}$ if the number $s$ does not matter, $\exists \vec{x}$ meaning $\exists x_1 \ldots \exists x_n$ and so on.
For practical reasons we transfer the terms \emph{first-order variable} and \emph{second-order variable} to the Boolean realm when referring to functions of arity zero resp. greater than zero.

\begin{definition}[Semantics of QBSF]
An \emph{interpretation} $\calI$ is a map from function variables $f^n$ to $n$-ary Boolean functions. A function variable occurs \emph{freely} if it is not in the scope of a matching quantifier. Write $\fr(\varphi)$ for all free variables in the qbsf $\varphi$. For $\calI$ that are defined on $\fr(\varphi)$, write $\llbracket \varphi \rrbracket_\calI$ for the valuation of $\varphi$ in $\calI$, which is defined as
\begin{alignat*}{2}
&\llbracket c \rrbracket_\calI &&\dfn c\text{ for }c \in \{0, 1\}\\
&\llbracket \varphi \land \psi \rrbracket_\calI &&\dfn \llbracket \varphi \rrbracket_\calI \cdot \llbracket \psi \rrbracket_\calI\\
&\llbracket \neg \varphi \rrbracket_\calI &&\dfn 1 - \llbracket \varphi \rrbracket_\calI\\
&\llbracket f^n(\varphi_1,\ldots,\varphi_n) \rrbracket_\calI &&\dfn \calI(f^n)(\llbracket \varphi_1 \rrbracket_\calI, \ldots, \llbracket \varphi_n \rrbracket_\calI)\\
&\llbracket \exists f^n \varphi\rrbracket_\calI &&\dfn \max\Set{\llbracket\varphi\rrbracket_{\calI[f^n \mapsto F]} | F \colon \{0,1\}^n \to \{0,1\}}
\end{alignat*}
where $\calI[f^n \mapsto F]$ is the interpretation \stt{} $\calI[f^n \mapsto F](f^n) = F$ and $\calI[f^n \mapsto F](g^m) = \calI(g^m)$ for $g^m \neq f^n$.
\end{definition}

Write $\calI \models \varphi$ for a qbsf $\varphi$ if $\calI$ is defined on $\fr(\varphi)$ and $\llbracket \varphi \rrbracket_\calI = 1$.
Say that $\varphi$ \emph{entails} $\psi$, $\varphi \models \psi$, if $\calI \models \varphi \Rightarrow \calI \models \psi$ for all interpretations $\calI$ which are defined on $\fr(\varphi) \cup \fr(\psi)$. If $\varphi \models \psi$ and $\psi \models \varphi$, then $\varphi$ and $\psi$ are called \emph{equivalent}, in symbols $\varphi \equiv \psi$.

\begin{lemma}\label{thm:invariant}
The set of interpretations satisfying a qbsf $\varphi$ is invariant under substitution of equivalent subformulas in $\varphi$.
\end{lemma}
\begin{proof}
Proven by simple induction.
\end{proof}

Write $\QBSF$ for the set of all qbsfs $\varphi$ for which $\emptyset \models \varphi$ holds, \ie, $\varphi$ is satisfied by the empty interpretation.

If in a formula $\varphi$ all quantifiers are at the beginning of $\varphi$, then it is in \emph{prenex form}.
A second-order qbf is \emph{simple} if all function symbols have only propositions as arguments.
It is in \emph{conjunctive normal form (CNF)} if it is in prenex form, simple and the matrix is in propositional CNF. Analogously define \emph{disjunctive normal form (DNF)}.

\begin{theorem}\label{thm:soqbf-in-aexp}
$\QBSF \in \AEXPPOLY$.
\end{theorem}
\begin{proof}
First transform $\varphi$ into prenex form $\varphi'$ in polynomial time.
Evaluate $\varphi'$ by alternating between existentially and universally states for each quantifier alternation, and guess and write down the truth tables for the quantified Boolean functions. These functions have arity at most $\size{\varphi}$, thus this whole step requires time $2^{\size{\varphi}}\cdot \size{\varphi'}$. Evaluate the matrix in deterministic exponential time by looking up the truth tables and accept if and only if it is true.
\end{proof}

\begin{theorem}\label{thm:soqbf-hardness}
$\QBSF$ in CNF or DNF is $\leqlogm$-hard for $\AEXPPOLY$.
\end{theorem}
\begin{proof}
Let $L \in \AEXPPOLY$, where $L \subseteq \Sigma^*$ for some alphabet $\Sigma$.
Let $\text{bin}(L) \dfn \Set{ \text{bin}(x) | x \in L }$, where $\text{bin}(\cdot)$ efficiently encodes words from $\Sigma^*$ over $\{0,1\}$. As $L \leqlogm \text{bin}(L)$, we only need to consider languages $L$ over $\{0,1\}$.

By \Cref{thm:alt_seq_poly} there is a polynomial $p$ and a deterministic oracle Turing machine $M$ with polynomial runtime such that
\begin{align*}
x \in L \Leftrightarrow \;&\Game_1 A_1 \subseteq \{0,1\}^{p(\size{x})} \; \ldots\; \Game_{p(\size{x})} A_{p(\size{x})} \subseteq \{0,1\}^{p(\size{x})}\\
&\text{\stt{}} \; M \text{ accepts }x\text{ with oracle }\langle A_1, \ldots, A_{p(\size{x})}\rangle\text{,}
\end{align*}
for an alternating quantifier sequence $\Game_1, \ldots, \Game_{p(\size{x})}$. Let $\ell = p(\size{x})$.

In this reduction we will represent the oracles $A_i$ as their characteristic Boolean functions, call them $c_i$.
Translate $M$ and $x$ into a formula $\exists \vec{z} \varphi_x(\vec{z})$, where $\size{\vec{z}}$ and $\size{\varphi_x}$ both are polynomial in $\size{x}$, and where $\varphi_x$ is in CNF. The encoding can be done as in \cite{cook_complexity_1971} and is possible in logspace (iterate over each possible timestep, tape, position and transition).

In an oracle-free setting we could now claim that $\exists \vec{z} \varphi_{x}(\vec{z})$ is true if and only if $M$ accepts $x$. The oracle questions queried in transitions to the state $q_?$ however require special handling. Assume that $M$ uses only $r + m$ tape cells for the oracle questions, $r \in \bigO{\log \ell}$, $m \in \bigO{\ell}$, \ie, it writes the index of the oracle and the concrete query always on the same cells. Let the proposition $z^t_{p} \in \vec{z}$ mean that at timestep $t$ on position $p$ of the oracle tape there is a one. Then modify $\varphi_{x}$ as follows. Let a clause $C$ of $\varphi_{x}$ encode a possible transition from some state $q$ to the state $q_{?}$ in timestep $t$. If the correct oracle answer is $q_{+}$, then $z^t_{1} \ldots z^t_{r}$ must represent some number $i$ in binary, $1 \leq i \leq \ell$, and $c_i(z^t_{r + 1},\ldots, z^t_{r + m}) = 1$ must hold. For the answer $q_-$ analogously $c_i(z^t_{r + 1},\ldots, z^t_{r + m}) = 0$. Therefore any transition to $q_?$ at a timestep $t$ must be encoded not in $C$ but instead in the new clauses $C^{+}_{1}, \ldots, C^{+}_{\ell}, C^{-}_{1}, \ldots, C^{-}_{\ell}$. Every such clause $C^{+}_{i}$/$C^{-}_{i}$ contains the same literals as $C$, but additionally says that the oracle number is $i$, and contains a single second-order atom of the form $c_i(\ldots)$ or $\neg c_i(\ldots)$. The new state of the transition is then obviously changed to $q_{+}$/$q_{-}$ instead of $q_{?}$. As $\ell$ is polynomial in $\size{x}$, there are also only polynomially many cases for the oracle number. The number of arguments of the characteristic functions is exactly $m$ which is again polynomial in $\size{x}$. The logspace-computability of the new clauses is straightforward.
Altogether the second-order qbf $\Game_1 c^m_1 \ldots \Game_\ell c^m_\ell \; \exists \vec{z} \;\, \varphi_x$ is true if and only if $x \in L$.

Let us now consider the DNF case. As $M$ is deterministic, there is another deterministic oracle machine $M'$ with identical runtime which simulates $M$ including the oracle calls, but then rejects any word that is accepted by $M$ and vice versa. Let the formula $\varphi_{x}'$ be the translation of $(M',x)$ as explained before. Then $x \in L$ iff $\Game_1 c_1 \ldots \Game_\ell c_\ell \; \neg \exists \vec{z} \;\, \varphi'_x$ is true iff $\Game_1 c_1 \ldots \Game_\ell c_\ell \; \forall \vec{z} \;\, \widehat{\varphi'_x}$ is true, where $\widehat{\varphi'_x}$ is the dual formula of $\varphi'_x$ (\ie, the negation normal form of its negation) and thus in DNF.
\end{proof}
\begin{corollary}
$\QBSF$ in CNF or DNF is $\leqlogm$-complete for $\AEXPPOLY$.
\end{corollary}
 
\section{Fragments with bounded quantifier alternation}

In Orponen's original characterization of the $\SigmaE{k}$ classes a language $A$ is expressed by a sequence of $k$ alternatingly quantified oracles, the input being verified by a $\PiP{k}$ oracle machine \cite{orponen}. Baier and Wagner improved this to a single word quantifier in the "first-order" suffix of the characterization instead of $k$ word quantifiers (see \Cref{thm:simulation}).

In this section we use a different strategy to reduce the first-order quantifier alternations directly on the level of formulas. The difference to Baier's and Wagner's result is that we obtain CNF formulas where previously only DNF formulas could be obtained and vice versa.
We define the following restricted problem of QBSF:

\begin{definition}
Let $n,m,k,\ell$ be non-negative integers, or $\omega$, and $P,Q \in \{ \Sigma, \Pi \}$. Write $\QBSF(P^n_m Q^\ell_k)$ for the restriction of $\QBSF$ to (prenex) formulas of the form
\[
\Game_1 f_1 \; \ldots\; \Game_p f_p \;\; \Game'_1 g^0_1 \;\ldots\; \Game'_q g^0_q \;\;H
\]
where $H$ is quantifier-free, $f_i$ are functions of arbitrary arities, $g_i$ are functions of arity zero (\ie, propositional variables), $p \leq n, q \leq \ell$, the quantifiers $\Game_1 \ldots \Game_p$ alternate at most $m-1$ times, the quantifiers $\Game'_1 \ldots \Game'_\ell$ alternate at most $k-1$ times, $\Game_1 = \exists$ iff $Q = \Sigma$, and $\Game'_{1} = \exists$ iff  $Q' = \Sigma$.
\end{definition}

\begin{example}
The formula $\exists f \, \forall x \, \forall y \, (x \land y \leftrightarrow f(x, y))$ is in $\QBSF(\Sigma^1_1 \Pi^2_1)$ and $\QBSF(\Sigma^\omega_1, \Pi^\omega_1)$, but not in $\QBSF(\Pi^1_1 \Sigma^\omega_\omega)$ or in $\QBSF(\Sigma^1_1 \Pi^1_1)$.
\end{example}

The following theorem demonstrates the reduction of propositional quantifier alternations.
The idea behind this is that the truth of the whole first-order part can be encoded in a single Boolean function. Denote dual quantifiers as $\overline{\exists} \dfn \forall$ and $\overline{\forall} \dfn \exists$.

\begin{theorem}[First-order alternation reduction]\label{thm:qbf_reduct}
Let $\varphi = \Game f \Game_1 x^0_1 \ldots \Game_k x^0_k H$ be a qbsf such that $H$ is quantifier-free and in CNF \emph{(}$\Game = \exists$\emph{)} resp. in DNF \emph{(}$\Game = \forall$\emph{)}.

Then there is an equivalent qbsf $\xi = \Game g \overline{\Game} x^0_1 \ldots \overline{\Game} x^0_k H'$ computable in logarithmic space, where $H'$ is in CNF \emph{(}$\Game = \exists$\emph{)} resp. in DNF \emph{(}$\Game = \forall$\emph{)}.
\end{theorem}

\begin{proof}
Let $f$ have arity $m$, this will be important later on. The first-order part of $\varphi$ is $\varphi' \dfn \Game_1 x_1 \ldots \Game_k x_k H$ (we drop the arities from now on) --- as all quantified variables are merely propositions, it is an ordinary qbf, except that function atoms occur in $H$. Let $\calI$ be some interpretation of $f$. To verify $\calI \models \varphi'$ we can use a set $S$ which models an \emph{assignment tree} of $\varphi'$. $S$ should have the following properties: It contains the empty assignment; further if $S$ contains some partial assignment $s$ to $x_1, \ldots, x_{j-1}$ and $\Game_j = \exists$ ($\forall$), then it must also contain $s \cup \{x_j \mapsto 0\}$ or (and) $s \cup \{x_j \mapsto 1\}$. For any total assignment $s \in S$, \ie, which is defined on all $x_1, \ldots, x_k$, the interpretation $\calI \cup s$ must satisfy $H$.  It is clear that, by the semantics of QBSF, such $S$ exists iff $\calI \models \varphi'$.

This set $S$ is encoded in a new quantified function $g$ with arity $m^* \dfn \max\{m, 2k\}$. We define how $g$ represents each (possibly partial) assignment $s \in S$. For each propositional variable $x_i$ we use two bits, the first one tells if $s(x_i)$ is defined (1 if yes, 0 if no), and the second one tells the value $s(x_i) \in \{0,1\}$ (and, say, 0 if undefined). It is $m^* \geq 2k$, so all bits will fit into the arguments of $g$, and if $g$ has larger arity than $2k$ then the trailing bits are assumed constant 0. Write $\langle s \rangle$ for the binary vector of length $m^*$ which encodes $s$, then $g(\langle s\rangle) = 1$ iff $s \in S$.
For the actual reduction of the quantifier rank consider two cases. In the case $\Game= \exists$ it is $H$ in CNF, say $H \dfn \bigwedge_{i = 1}^{n} C_i$ for clauses $C_i$. The conditions of the set $S$ encoded by a given $g$ are verified by the following formula in CNF:
\begin{align*}
 \vartheta_{\varphi'}(g) \dfn &\forall x_1 \ldots \forall x_k \quad g(\vec{0}) \, \land \, \bigwedge_{i=1}^{n} \Big(g(1, x_1, \ldots, 1, x_k, \vec{0}) \rightarrow C_i\Big) \; \land\\
  \quad \bigwedge_{\substack{i=1\\ \Game_i = \exists}}^{k-1} &\Bigg(g(1, x_1, \ldots, 1, x_i, \vec{0}) \rightarrow\\
   \quad &\qquad\big(g(1, x_1, \ldots, 1, x_i, 1, 0, \vec{0}) \lor g(1, x_1, \ldots, 1, x_i, 1, 1, \vec{0})\big)\Bigg) \; \land\\
 \quad\bigwedge_{\substack{i=1\\ \Game_i = \forall}}^{k-1} &\Bigg(g(1, x_1, \ldots, 1, x_i, \vec{0}) \rightarrow g(1,x_1, \ldots, 1, x_i, 1, 1, \vec{0})\Bigg)\;  \land\\
 &\quad\Bigg(g(1, x_1, \ldots, 1, x_i, \vec{0}) \rightarrow g(1, x_1, \ldots, 1, x_i, 1, 0, \vec{0})\Bigg)
\end{align*}
$\vartheta_{\varphi'}(g)$ is logspace-computable from $\varphi'$. In $\varphi$ now replace $\varphi'$ with $\exists g \; \vartheta_{\varphi'}$. To see the correctness of this step assume that $\calI$ is an interpretation of $x_1,\ldots,x_k$. Since $\calI \models \varphi' \Leftrightarrow \calI \models \exists g\; \vartheta_{\varphi'}$, as explained above, we can apply \Cref{thm:invariant}.

For the case $\Game = \forall$ it is $\varphi'$ is in DNF. Consider $\vartheta_{\psi}$ where $\psi$ is the dual of $\varphi'$. Note that $\psi$ itself has a matrix in CNF and $\vartheta_{\psi}$ thus can be constructed as above. Further it holds
\[
\calI \models \varphi' \Leftrightarrow \calI \not\models \psi \Leftrightarrow \calI \not\models \exists g \; \vartheta_{\psi} \Leftrightarrow \calI \models  \forall g \; \neg \vartheta_{\psi}\text{.}
\]
Therefore replace $\varphi'$ now with $\forall g \; \widehat{\vartheta}_{\psi}$, where $\widehat{\vartheta}_{\psi}$ is the dual of $\vartheta_{\psi}$, and hence again in DNF.

Finally replace all occurrences of $f(a_1, \ldots, a_m)$ by $f(a_1, \ldots, a_m, 0, \ldots, 0)$, \ie, pad any possible interpretation of $f$ with zeros up to arity $m^*$. The functions $g$ and $f$ have then the same arity and identical quantifier type $\Game$. Hence we can merge them into a single function: Replace $\Game f \Game g$ by $\Game h$, and as well each expression $f(a_1, \ldots, a_{m^*})$ in the matrix with $h(0, a_1, \ldots, a_{m^*})$ and likewise $g(a_1, \ldots, a_{m^*})$ with $h(1, a_1, \ldots, a_{m^*})$. It is easy to see that the matrix then holds for some (all) interpretation(s) of $h$ if and only if it holds for some (all) interpretation(s) of $f$ and $g$. This concludes the proof.
\end{proof}

\begin{theorem}\label{thm:no-alt-completeness}
The following problems restricted to CNF or DNF are $\leqlogm$-complete:
\begin{itemize}
\item If $k$ is even, then $\QBSF(\Sigma^k_k \Sigma^{\omega}_1)$ for $\SigmaE{k}$ and $\QBSF(\Pi^k_k \Pi^{\omega}_1)$ for $\PiE{k}$.
\item If $k$ is odd, then $\QBSF(\Sigma^k_k \Pi^{\omega}_1)$ for $\SigmaE{k}$ and $\QBSF(\Pi^k_k \Sigma^{\omega}_1)$ for $\PiE{k}$.
\item $\QBSF(\Sigma^\omega_\omega \Sigma^{\omega}_1)$ and $\QBSF(\Sigma^\omega_\omega \Pi^{\omega}_1)$  for $\AEXPPOLY$.
\end{itemize}
\end{theorem}
\begin{proof}
The upper bounds work as in \Cref{thm:soqbf-in-aexp} by guessing the truth tables of the quantified functions.

For the lower bound consider a reduction similar to the proof of \Cref{thm:soqbf-hardness}. By \Cref{thm:simulation}  we can already correctly choose the number and quantifier type of the functions $c_1,\ldots,c_k$ when reducing from a $\SigmaE{k}$ or $\PiE{k}$ language. The first-order part can be constructed accordingly as in \Cref{thm:soqbf-hardness}, but, due to the single word quantifier introduced in \Cref{thm:simulation}, has now the form $\exists \vec{y} \, \exists \vec{z}  \,\varphi'_{x}(\vec{y},\vec{z})$ in CNF or  $\forall \vec{y} \, \forall \vec{z} \, \varphi'_{x}(\vec{y},\vec{z})$ in DNF. Note that we can represent the computation of the deterministic machine on input $x$ arbitrarily as $\exists \vec{z}  \,\varphi'_{x}$ in CNF or $\forall \vec{z}  \,\varphi'_{x}$ in DNF, therefore choose the quantifiers matching as stated above.

By this construction, the hardness for the CNF cases with $\Sigma^\omega_1$ first-order part and the DNF cases with $\Pi^\omega_1$ first-order part is shown.
In the remaining cases apply \Cref{thm:qbf_reduct} to obtain an equivalent formula but with CNF matrix after an $\Pi^\omega_1$ first-order prefix resp. DNF matrix after an $\Sigma^\omega_1$ first-order prefix.\end{proof}
 
\section{Fragments with Skolem functions}

In the previous sections we considered the QBSF problem where function atoms could occur multiple times in a formula, in particular with different arguments. A Skolem function of a proposition $x$ however is a Boolean function that depends only on certain other propositions $y_1,\ldots,y_n$, the so-called \emph{dependencies} of $x$. Hence, to connect QBSF to the Dependency QBF problem \cite{dqbf} and other logics of independence, we now focus on formulas where all quantified functions are Skolem functions:

\begin{definition}
Let $n,m,k,\ell$ be non-negative integers or $\omega$. Let $P,Q \in \{\Sigma,\Pi\}$. Write $\SKOL(P^n_m Q^\ell_k)$ for the restriction of $\QBSF(P^n_m Q^\ell_k)$ to formulas in which for all function symbols $f$ it holds that $f$ always occurs with the same arguments.
\end{definition}

In contrast, \textsf{DQBF} is defined as follows:

\begin{definition}
Every formula of the form $\forall \vec{x} \; \exists y_1(\vec{z}_1) \ldots \exists y_n(\vec{z}_n) \; H$ is called a \emph{dqbf}, where the matrix $H$ is a quantifier-free propositional formula and $\vec{z}_i \subseteq \vec{x}$ \fa $i = 1,\ldots,n$.

A dqbf of this form is \emph{true} if for all $i = 1,\ldots,n$ there is a Skolem function $Y_i$ of $y_i$ depending only on $\vec{z}_i$ \suchthat for all assignments to $\vec{x}$ the matrix $H$ evaluates to true, provided the values of $Y_i$ are assigned to $y_i$.

As a decision problem, \textsf{DQBF} is defined as the set of all true dqbfs.
\end{definition}

\smallskip

In this section we will prove that the restricted problem $\SKOL$ is complete for the same complexity classes as the general case, \ie, the number of function quantifier alternations again determines the level in the exponential hierarchy.

Orponen's characterization via alternating oracle quantification allows polynomial time machines to recognize languages with exponential time complexity \cite{orponen}. \citeauthor{hannula_complexity_2015} generalized this to handle a polynomial number of oracles \cite{hannula_complexity_2015}. In the following we use the notion of \emph{tableaus} to adapt these characterization to our needs.

Call an oracle machine \emph{single-query machine} if it asks at most one oracle question. An indirect simulation with a single-query machine allows, since the oracle tape cells directly correspond to the arguments of the quantified functions, an encoding in $\SKOL$ as follows.

\medskip

If $M$ is an ATM with state set $Q$ and tape alphabet $\Gamma$, then a \emph{configuration} of $M$ is a finite sequence $C \in (Q \cup \Gamma)^*$ which contains exactly one state.
A \emph{tableau} $T$ of $M$ is a finite sequence $C_1,\ldots,C_n$ of equally long configurations of $M$ such that each $C_{i+1}$ results from $C_i$ by a transition of $M$.
A tableau is \emph{pure} if all states assumed in the tableau except the last one have the same alternation type, \ie, all existential or all universal. A pure tableau $C_1,\ldots,C_n$ is \emph{alternating} if $n \geq 2$ and $q_n$ has a different alternation type than $q_1,\ldots,q_{n-1}$, where $q_i$ is the state of $C_i$.
If $T, T'$ are tableaus of $M$, then $T'$ is a \emph{successor tableau} of $T$ if its first configuration is equal to $T$'s final configuration.

\smallskip

Let $T$ be a pure tableau that assumes $q$ as its last state. Say that $T$ is $k$-\emph{accepting} if
\begin{itemize}
	\item either $k = 1$ and $q$ is an accepting state of the machine $M$,
	\item or $k > 1$ and $T$ is alternating and
	\begin{itemize}
	\item $q$ is existential and $T$ has a pure $(k-1)$-accepting successor tableau,
	\item or $q$ is universal and all pure successor tableaus of $T$ are $(k-1)$-accepting.
	\end{itemize}
\end{itemize}

\begin{theorem}[Single-query indirect simulation]\label{thm:single-query}
For every $Q \in \{\Sigma, \Pi\}$ and every $Q^E_{g(n)}$-machine $M$ there is a polynomial $h$ and a single-query oracle $\SigmaP{4}$-machine $N$ such that $M$ accepts $x$ if and only if
\[
\Game_1 A_1 \; \Game_2 A_2  \; \ldots \; \Game_{g(\size{x})} A_{g(\size{x})} \,:\, N \text{ accepts } x \text{ with oracle }\langle A_1, \ldots, A_{g(\size{x})} \rangle\text{,}
\]
where $\Game_1, \ldots$ are alternating quantifiers starting with $\exists$ ($Q = \Sigma$) resp. $\forall$ ($Q = \Pi$), and further $A_i \subseteq \{0,1\}^{h(\size{x})}$ \fa $i = 1, \ldots, g(\size{x})$.
\end{theorem}
\begin{proof}
Let $M$ have runtime $f$, and let $m \dfn g(\size{x})$ and $n \dfn f(\size{x})$.
\Wloss{} we can assume the following: $m \leq n$, $2 \leq n$, $M$ does always exactly $m$ alternations before it accepts or rejects, and in each alternation phase it does exactly $n$ steps before it alternates, rejects or accepts. These properties imply that $M$ accepts $x$ if and only if all resp. some of its pure tableau starting with the initial configuration are $m$-accepting.
We further assume for simplicity that $M$ uses only one tape. 
\smallskip

The idea is to encode the tableaus $T_1, \ldots, T_m$ of the alternation phases of the computation of $M$ in the oracles $A_1, \ldots, A_m$ as follows. Words of $A_i$ are \emph{indexed cells} $w = (c, t, p)$. Here $p$ denotes the position of the cell on the tape, $t$ is the current timestep, and $c \in Q \cup \Gamma$ is either the symbol written at position $p$, or the state of $M$ if $p$ happens to be the head position at timestep $t$.
Let $h$ be the polynomial size of such a window $w$ encoded over $\{0,1\}$ (by binary encoding of $t$ and $p$).
A set $A \subseteq \{0,1\}^{h(\size{x})}$ then represents a tableau $T = C_1, \ldots, C_n$ in the sense that $(c, t, p) \in A$ if and only if the cell at tape position $p$ contains $c$ at timestep $t$.

\smallskip

We now translate the acceptance condition of $M$ into a formal expression, using $k$-acceptance of tableaus:
\[
\mathrm{Acc}_i \dfn \begin{cases}\exists A_i \subseteq \{ 0,1\}^{h(\size{x})} \quad(\mathrm{Val}_i \land \mathrm{Init}_i \land \mathrm{Alt}_i \land \mathrm{Acc}_{i+1})\quad &\text{if }\Game_i = \exists\\
\forall A_i \subseteq \{0,1\}^{h(\size{x})} \quad(\mathrm{Val}_i \land \mathrm{Init}_i) \rightarrow (\mathrm{Alt}_i \land \mathrm{Acc}_{i+1})\quad &\text{if }\Game_i=\forall\end{cases}
\]
for $1 \leq i \leq m$. The semantics is that $\mathrm{Val}_i$ is true if $A_i$ encodes a pure tableau of $M$, $\mathrm{Init}_1$ is true if there is the initial configuration of $M$ on $x$ encoded in $A_1$,  $\mathrm{Init}_i$ for $i > 1$ is true if the first configuration of $A_i$ is equal to the last configuration of $A_{i-1}$, (\ie, $A_i$ is a successor tableau of $A_{i-1}$), $\mathrm{Alt}_i$ for $i < m$ is true if the tableau encoded in $A_i$ is end-alternating, $\mathrm{Alt}_m$ is true if the tableau encoded in $A_m$ is $1$-accepting, and $\mathrm{Acc}_{m+1}$ is always true.

\smallskip

By the above definitions $M$ accepts $x$ if and only if the predicate $\mathrm{Acc}_1$ is true, as it states that $M$ has an $m$-accepting pure initial tableau.
The formula can be written in prenex form, \ie, $\mathrm{Acc}_1$ holds if and only if $\Game_1 A_1 \subseteq \{0,1\}^{h(\size{x})} \ldots \Game_m A_m \subseteq \{0,1\}^{h(\size{x})} \; V_1$ holds, where the predicate $V_i$ is defined as
\[
V_i \dfn \begin{cases}(\mathrm{Val}_i \land \mathrm{Init}_i \land \mathrm{Alt}_i \land V_{i+1})&\text{if }\Game_i = \exists\\
(\mathrm{Val}_i \land \mathrm{Init}_i) \rightarrow (\mathrm{Alt}_i \land V_{i+1})\quad &\text{if }\Game_i =\forall\end{cases}
\]
for $1 \leq i \leq m$, and $V_{m+1} = 1$.
To prove the theorem we give now a single-query oracle ATM $N$ with polynomial runtime and $4$ alternations which accepts if and only if $V_1$ is true.

\smallskip

For a predicate $P$ write $\overline{P}$ for its complement. Group the predicates above as follows:
\begin{align*}
\calT_i &\dfn \Set{ \mathrm{Val}_j, \mathrm{Init}_j, \mathrm{Alt}_j  | 1 \leq j < i} \text{ for }i = 1, \ldots, m+1\text{, }\\
\calF^0_i &\dfn \Set{ \overline{\mathrm{Val}_i} }\text{, }\calF^1_i \dfn \Set{ \overline{\mathrm{Init}_i} }\text{ for }i=1,\ldots, m\text{, and }\calF^0_{m+1} \dfn \emptyset\text{, } \calF^1_{m+1} \dfn \emptyset\text{.}
\end{align*}
By its definition $V_1$ is true if and only if $\exists i \in \{1,\ldots,m\}, \exists d \in \{0,1\}$ \suchthat all predicates in $\calS^d_i \dfn \calT_i \cup \calF^d_i$ are true and further $\Game_i = \forall$ or $i > m$.
Hence the machine $N$ is defined to work as follows:
\begin{enumerate}
	\item In time $\bigO{\log g}$ existentially guess $i$ and $d$,
	\item In time $\bigO{\log g}$ universally branch on every predicate $P$ in $\calS^d_i$,
	\item Verify that $P$ is true.
\end{enumerate}

It only remains to verify that every predicate $P$ (and accordingly $\overline{P}$) in $\calS^d_i$ can be checked in polynomial time, with only one oracle query, and at most two additional alternations.

We sketch the required alternating procedures, where quantifier symbols $\exists, \forall$ always imply branching.

If $(c,t,p)$ is a cell, then $w \in A_i$ means that $A_i$ contains the encoding of $w$. The available timesteps $t$ in each tableau range over $\{0, \ldots, n\}$. The available positions $p$ in the configurations are $\{0,\ldots,n,n+1,\ldots,2n+1\}$, where the input word is placed on positions $n+1, \ldots, n + \size{x}$ and the initial state of $M$ is given on position $n$.

The predicates are checked as follows:
\begin{itemize}
	\item $\mathrm{Val}_i$: (check in parallel)\begin{itemize}
	\item $\forall w \in \{0,1\}^h$ : if $w$ is no valid encoded cell then $w \notin A_i$,
	\item $\forall t  \, \forall p \; \exists c \in Q \cup \Gamma \; : \;(c,t,p) \in A_i$,
	\item $\forall w_0 = (c,t,p) \; \forall w_1 = (c', t, p) : $ if $c \neq c'$ then $\exists j \in \{0,1\}$ \suchthat $w_j \notin A_i$,
	\item $\forall w_0 = (c_0, t, p - 1) \;\, \forall w_1 = (c_1, t, p) \;\, \forall w_2 = (c_2, t, p + 1)$\\
    $\forall w_3 = (c_3, t + 1, p -1)\;\,  \forall w_4 = (c_4, t+1, p)\;\,  \forall w_5 = (c_5, t+1, p+1)$ :\\
    if $M$ has no transition from $(w_0,w_1,w_2)$ to $(w_3,w_4,w_5)$ then $\exists j \in \{0,\ldots,5\}$ \suchthat $w_j \notin A_i$,
	\item $\forall w_0 = (c,t,p) \; \forall w_1 =(c', t', p') :$ if $t < t' < n$ and $c, c'$ are states with different alternation types then $\exists j \in \{0,1\}$ \suchthat $w_j \notin A_i$,
	\end{itemize}
	\item $\mathrm{Alt}_i$, $i < m$: $\exists w_0 = (c, n-1,p) \; \exists w_1 = (c', n, p')$ \suchthat $w_0$ and $w_1$ contain states with different alternation types and $\forall j \in \{0,1\} : w_j \in A_i$,
	\item $\mathrm{Alt}_{m}$: $\exists w = (q,n,p)$ : $q$ is an accepting state of $M$ and $w \in A_m$
	\item $\mathrm{Init}_1$: (check in parallel)\begin{itemize}
	\item $\forall i \in \{1, \ldots, \size{x}\} \; \exists w = (c, 0, n+i)$ \suchthat $c$ is the $i$-th letter of $x$ and $w \in A_1$,
    \item $(q_0,0,n) \in A_1$ where $q_0$ is the initial sate of $M$,
    \item $\forall i \notin \{n, \ldots,n+\size{x}\} \; \exists w = (\Box, 0, i) \in A_1$,
	\end{itemize}
	\item $\mathrm{Init}_i$, $i > 1$: $\forall w_0 = (c, 0, p) \; \exists w_1 = (c', n, p)$ : if $c \neq c'$ then $\exists j \in \{0,1\}$ \suchthat $w_j \notin A_{i-j}$.\qedhere
\end{itemize}\end{proof}

\begin{theorem}
For $k \geq 1$, the following problems restricted to DNF are $\leqlogm$-complete:
\begin{itemize}
	\item $\SKOL(\Sigma^k_k \Sigma^{\omega}_4)$ for $\SigmaE{k}$,
	\item $\SKOL(\Pi^k_k \Sigma^{\omega}_4)$ for $\PiE{k}$,
	\item $\SKOL(\Sigma^\omega_\omega \Sigma^{\omega}_4)$ for $\AEXPPOLY$.
\end{itemize}
\end{theorem}
\begin{proof}
    The proof of the upper bounds is essentially the same as for \Cref{thm:no-alt-completeness}. The lower bound proof is again similar to \Cref{thm:soqbf-hardness}. It holds that the translation from $\SigmaP{k}$ machines to deterministic machines with word quantifiers relativizes (see \textcite[Lem.\ 1.1]{baker_second_1979}) and additionally preserves the single-query property.

	As only one oracle question is asked by the encoded machine (at timestep, say, $t$), every function symbol occurs only with a fixed argument set, which describes the content of the oracle tape (modulo the oracle index) at timestep $t$.
\end{proof}
The method of single-query indirect simulation can be applied to obtain an alternative proof for the hardness of \textsf{DQBF}. Peterson, Reif and Azhar \cite{dqbf} state that every dqbf has an equivalent \emph{functional form} which is in essence a QBSF formula with implicit function symbols. Similarly, all $\SKOL(\Sigma_1^\omega\Sigma^\omega_\omega)$ formulas are equivalent to the functional form of a DQBF formula with a straightforward efficient translation.

\begin{corollary}
	\textsf{DQBF} is $\leqlogm$-complete for $\SigmaE{1}$.
\end{corollary}
 
\section{Conclusion}

The presented completeness results for the exponential hierarchy are in analogy to the results known for QBF; still they differ in subtle points.
One difference is that the "$\omega$-jump" of QBSF is complete for $\AEXPPOLY$ and not for $\EXPSPACE$. The reason for this is that any given input of length $n$ with explicit quantifiers, like in the QBF style, can only express $n$ alternations. This differs from  decision problems which are defined via exponentially many alternations, \eg, certain games. It may be that the class $\AEXPPOLY$ is perhaps a more natural analogy to $\AP$ than $\AEXP$, at least in the cases where the number of quantifiers is bounded by the input itself, \eg, logical operators.

\smallskip

Other differences are with regard to normal forms: The $\SigmaP{k}$-hardness of \textsf{QBF}$_k$ already holds for CNF --- but only if the rightmost quantifier happens to be existential, \ie, $k$ is odd. If it is universal, \ie, $k$ is even, then \textsf{CNF-QBF}$_k$ is in $\SigmaP{k-1}$, \ie, it is supposedly easier. On the other hand, DNF establishes hardness for $\SigmaP{k}$ only for even $k$.

This collapse however does not occur in $\QBSF$.
This peculiar robustness can be explained if one remembers that function symbols occur in formulas. While the innermost existential guessing can be avoided in propositional formula in DNF (just scan for a single non-contradicting conjunction), this is not possible here: Is the conjunction $f(x_{1},x_{2}) \land g(x_{2},x_{3})$ self-contradicting or not? The hardness results are a hint that the structure of formulas with Boolean second-order variables is unlikely to exhibit such shortcuts as in propositional logic.

\smallskip

On the other hand for $\SKOL$ no such symmetry of CNF and DNF could be established, as \Cref{thm:qbf_reduct} does not preserve the $\SKOL$ condition.
Similarly, the proof of \Cref{thm:qbf_reduct} can at best produce 3CNF (and non-Horn) formulas, even if the matrix of the formula was already in 2CNF and Horn form. How does the complexity of $\QBSF$ change if 2CNF or 2DNF is considered, or horn formulas? How do CNF and DNF influence the complexity of the $\SKOL$ fragment? Can the single-query indirect simulation be done by a $\SigmaP{k}$ machine with $k < 4$? Can the problem \textsf{DQBF} be generalized to incorporate universal quantification of Skolem functions?

\subsection*{Acknowledgement}

The author thanks Heribert Vollmer for helpful discussions and hints as well as the anonymous referees for spotting errors and improving the clarity of this paper.

\printbibliography{}

\end{document}